\documentclass[amsart]{article}
\usepackage{bm}
\usepackage{tikz}
\usepackage{tikz,pgfplots}
\usepackage{caption}
\usepackage{color}
\usepackage{amsmath}
\usepackage{amssymb}\usepackage{amsthm}
\usepackage{amsfonts}
\usepackage{amscd}
\begin{document}
\theoremstyle{plain}
\newtheorem*{ithm}{Theorem}
\newtheorem*{iprop}{Proposition}
\newtheorem*{idefn}{Definition}
\newtheorem{thm}{Theorem}[section]
\newtheorem{lem}[thm]{Lemma}
\newtheorem{dlem}[thm]{Lemma/Definition}
\newtheorem{prop}[thm]{Proposition}
\newtheorem{set}[thm]{Setting}
\newtheorem{cor}[thm]{Corollary}
\newtheorem*{icor}{Corollary}
\theoremstyle{definition}
\newtheorem{assum}[thm]{Assumption}
\newtheorem{notation}[thm]{Notation}
\newtheorem{setting}[thm]{Setting}
\newtheorem{defn}[thm]{Definition}
\newtheorem{clm}[thm]{Claim}
\newtheorem{ex}[thm]{Example}
\theoremstyle{remark}
\newtheorem{rem}[thm]{Remark}
\newcommand{\unit}{\mathbb I}
\newcommand{\ali}[1]{{\caB}_{[ #1 ,\infty)}}
\newcommand{\alm}[1]{{\caB}_{(-\infty, #1 ]}}
\newcommand{\nn}[1]{\lV #1 \rV}
\newcommand{\br}{{\mathbb R}}
\newcommand{\dm}{{\rm dom}\mu}
\newcommand{\lb}{l_{\bb}(n,n_0,k_R,k_L,\lal,\bbD,\bbG,Y)}
\newcommand{\Ad}{\mathop{\mathrm{Ad}}\nolimits}
\newcommand{\Proj}{\mathop{\mathrm{Proj}}\nolimits}
\newcommand{\RRe}{\mathop{\mathrm{Re}}\nolimits}
\newcommand{\RIm}{\mathop{\mathrm{Im}}\nolimits}
\newcommand{\Wo}{\mathop{\mathrm{Wo}}\nolimits}
\newcommand{\Prim}{\mathop{\mathrm{Prim}_1}\nolimits}
\newcommand{\Primz}{\mathop{\mathrm{Prim}}\nolimits}
\newcommand{\ClassA}{\mathop{\mathrm{ClassA}}\nolimits}
\newcommand{\Class}{\mathop{\mathrm{Class}}\nolimits}
\newcommand{\diam}{\mathop{\mathrm{diam}}\nolimits}
\def\qed{{\unskip\nobreak\hfil\penalty50
\hskip2em\hbox{}\nobreak\hfil$\square$
\parfillskip=0pt \finalhyphendemerits=0\par}\medskip}
\def\proof{\trivlist \item[\hskip \labelsep{\bf Proof.\ }]}
\def\endproof{\null\hfill\qed\endtrivlist\noindent}
\def\proofof[#1]{\trivlist \item[\hskip \labelsep{\bf Proof of #1.\ }]}
\def\endproofof{\null\hfill\qed\endtrivlist\noindent}
\numberwithin{equation}{section}

\newcommand{\kakunin}[1]{{\color{red}}}

\newcommand{\oo}{{\boldsymbol\omega}}
\newcommand{\ctv}{\caC_{(\theta,\varphi)}}
\newcommand{\btv}{\caB_{(\theta,\varphi)}}
\newcommand{\amf}{\caB}
\newcommand{\at}{\caA_{\bbZ^2}}
\newcommand{\oz}{\caO_0}
\newcommand{\caA}{{\mathcal A}}
\newcommand{\caB}{{\mathcal B}}
\newcommand{\caC}{{\mathcal C}}
\newcommand{\caD}{{\mathcal D}}
\newcommand{\caE}{{\mathcal E}}
\newcommand{\caF}{{\mathcal F}}
\newcommand{\caG}{{\mathcal G}}
\newcommand{\caH}{{\mathcal H}}
\newcommand{\caI}{{\mathcal I}}
\newcommand{\caJ}{{\mathcal J}}
\newcommand{\caK}{{\mathcal K}}
\newcommand{\caL}{{\mathcal L}}
\newcommand{\caM}{{\mathcal M}}
\newcommand{\caN}{{\mathcal N}}
\newcommand{\caO}{{\mathcal O}}
\newcommand{\caP}{{\mathcal P}}
\newcommand{\caQ}{{\mathcal Q}}
\newcommand{\caR}{{\mathcal R}}
\newcommand{\caS}{{\mathcal S}}
\newcommand{\caT}{{\mathcal T}}
\newcommand{\caU}{{\mathcal U}}
\newcommand{\caV}{{\mathcal V}}
\newcommand{\caW}{{\mathcal W}}
\newcommand{\caX}{{\mathcal X}}
\newcommand{\caY}{{\mathcal Y}}
\newcommand{\caZ}{{\mathcal Z}}
\newcommand{\bba}{{\mathbb a}}
\newcommand{\bbA}{{\mathbb A}}
\newcommand{\bbB}{{\mathbb B}}
\newcommand{\bbC}{{\mathbb C}}
\newcommand{\bbD}{{\mathbb D}}
\newcommand{\bbE}{{\mathbb E}}
\newcommand{\bbF}{{\mathbb F}}
\newcommand{\bbG}{{\mathbb G}}
\newcommand{\bbH}{{\mathbb H}}
\newcommand{\bbI}{{\mathbb I}}
\newcommand{\bbJ}{{\mathbb J}}
\newcommand{\bbK}{{\mathbb K}}
\newcommand{\bbL}{{\mathbb L}}
\newcommand{\bbM}{{\mathbb M}}
\newcommand{\bbN}{{\mathbb N}}
\newcommand{\bbO}{{\mathbb O}}
\newcommand{\bbP}{{\mathbb P}}
\newcommand{\bbQ}{{\mathbb Q}}
\newcommand{\bbR}{{\mathbb R}}
\newcommand{\bbS}{{\mathbb S}}
\newcommand{\bbT}{{\mathbb T}}
\newcommand{\bbU}{{\mathbb U}}
\newcommand{\bbV}{{\mathbb V}}
\newcommand{\bbW}{{\mathbb W}}
\newcommand{\bbX}{{\mathbb X}}
\newcommand{\bbY}{{\mathbb Y}}
\newcommand{\bbZ}{{\mathbb Z}}
\newcommand{\str}{^*}
\newcommand{\lv}{\left \vert}
\newcommand{\rv}{\right \vert}
\newcommand{\lV}{\left \Vert}
\newcommand{\rV}{\right \Vert}
\newcommand{\la}{\left \langle}
\newcommand{\ra}{\right \rangle}
\newcommand{\ltm}{\left \{}
\newcommand{\rtm}{\right \}}
\newcommand{\lcm}{\left [}
\newcommand{\rcm}{\right ]}
\newcommand{\ket}[1]{\lv #1 \ra}
\newcommand{\bra}[1]{\la #1 \rv}
\newcommand{\lmk}{\left (}
\newcommand{\rmk}{\right )}
\newcommand{\al}{{\mathcal A}}
\newcommand{\md}{M_d({\mathbb C})}
\newcommand{\ainn}{\mathop{\mathrm{AInn}}\nolimits}
\newcommand{\id}{\mathop{\mathrm{id}}\nolimits}
\newcommand{\Tr}{\mathop{\mathrm{Tr}}\nolimits}
\newcommand{\Ran}{\mathop{\mathrm{Ran}}\nolimits}
\newcommand{\Ker}{\mathop{\mathrm{Ker}}\nolimits}
\newcommand{\Aut}{\mathop{\mathrm{Aut}}\nolimits}
\newcommand{\spn}{\mathop{\mathrm{span}}\nolimits}
\newcommand{\Mat}{\mathop{\mathrm{M}}\nolimits}
\newcommand{\UT}{\mathop{\mathrm{UT}}\nolimits}
\newcommand{\DT}{\mathop{\mathrm{DT}}\nolimits}
\newcommand{\GL}{\mathop{\mathrm{GL}}\nolimits}
\newcommand{\spa}{\mathop{\mathrm{span}}\nolimits}
\newcommand{\supp}{\mathop{\mathrm{supp}}\nolimits}
\newcommand{\rank}{\mathop{\mathrm{rank}}\nolimits}
\newcommand{\idd}{\mathop{\mathrm{id}}\nolimits}
\newcommand{\ran}{\mathop{\mathrm{Ran}}\nolimits}
\newcommand{\dr}{ \mathop{\mathrm{d}_{{\mathbb R}^k}}\nolimits} 
\newcommand{\dc}{ \mathop{\mathrm{d}_{\cc}}\nolimits} \newcommand{\drr}{ \mathop{\mathrm{d}_{\rr}}\nolimits} 
\newcommand{\zin}{\mathbb{Z}}
\newcommand{\rr}{\mathbb{R}}
\newcommand{\cc}{\mathbb{C}}
\newcommand{\ww}{\mathbb{W}}
\newcommand{\nan}{\mathbb{N}}\newcommand{\bb}{\mathbb{B}}
\newcommand{\aaa}{\mathbb{A}}\newcommand{\ee}{\mathbb{E}}
\newcommand{\pp}{\mathbb{P}}
\newcommand{\wks}{\mathop{\mathrm{wk^*-}}\nolimits}
\newcommand{\mk}{{\Mat_k}}
\newcommand{\mnz}{\Mat_{n_0}}
\newcommand{\mn}{\Mat_{n}}
\newcommand{\dist}{\dc}
\newcommand{\braket}[2]{\left\langle#1,#2\right\rangle}
\newcommand{\ketbra}[2]{\left\vert #1\right \rangle \left\langle #2\right\vert}
\newcommand{\abs}[1]{\left\vert#1\right\vert}
\newcommand{\trl}[2]
{T_{#1}^{(\theta,\varphi), \Lambda_{#2},\bar V_{#1,\Lambda_{#2}}}}
\newcommand{\trlz}[1]
{T_{#1}^{(\theta,\varphi), \Lambda_{0},\unit}}
\newcommand{\trlt}[2]
{T_{#1}^{(\theta,\varphi), \Lambda_{#2}+t_{#2}\bm e_{\Lambda_{#2}},\bar V_{#1,\Lambda_{#2}+t_{#2}\bm e_{\Lambda_{#2}}}}}
\newcommand{\trltj}[4]
{T_{#1}^{(\theta,\varphi), \Lambda_{#2}^{(#3)}+t_{#4}\bm e_{\Lambda_{#2}^{(#3)}},\bar V_{#1,\Lambda_{#2}^{(#3)}+t_{#4}\bm e_{\Lambda_{#2}^{(#3)}}}}}
\newcommand{\trltjp}[4]
{T_{#1}^{(\theta,\varphi), {\Lambda'}_{#2}^{(#3)}+t_{#4}'\bm e_{{\Lambda'}_{#2}^{(#3)}},\bar V_{#1,{\Lambda'}_{#2}^{(#3)}+t_{#4}'\bm e_{{\Lambda'}_{#2}^{(#3)}}}}}
\newcommand{\trlta}[2]
{T_{#1}^{(\theta,\varphi), \Lambda_{#2}^{t_{#2}},\bar V_{#1,\Lambda_{#2}^{t_{#2}}}}}
\newcommand{\trltb}[2]
{T_{#1}^{(\theta,\varphi), \Lambda_{#2}+t\bm e_{\Lambda_{#2}},\bar V_{#1,\Lambda_{#2}+t\bm e_{\Lambda_{#2}}}}}

\newcommand{\trlpt}[2]
{T_{#1}^{(\theta,\varphi), \Lambda_{#2}'+t_{#2}'\bm e_{\Lambda_{#2}'},\bar V_{#1,\Lambda_{#2}'+t_{#2}'\bm e_{\Lambda_{#2}'}}}}

\newcommand{\trll}[3]
{T_{#1, #3}^{(\theta,\varphi), \Lambda_{#2},\bar V_{#1,\Lambda_{#2}}}}
\newcommand{\trlp}[2]
{T_{#1}^{(\theta,\varphi), \Lambda_{#2}',\bar V_{#1,\Lambda_{#2}'}}}
\newcommand{\trlpp}[2]
{T_{#1}^{(\theta,\varphi), \Lambda_{#2}'',\bar V_{#1,\Lambda_{#2}''}}}
\newcommand{\trlj}[3]
{T_{\rho_{#1}}^{(\theta,\varphi), \Lambda_{#2}^{(#3)}, V_{\rho_{#1},\Lambda_{#2}^{(#3)}}}}
\newcommand{\trljp}[3]
{T_{{\rho'}_{#1}}^{(\theta,\varphi), {\Lambda'}_{#2}^{(#3)},V_{\rho'_{#1},{\Lambda'}_{#2}^{(#3)}}}}
\newcommand{\wod}[3]
{W_{#1\Lambda_{#2}\Lambda_{#3}}}
\newcommand{\wodt}[3]
{{W^{\bm t}}_{#1\Lambda_{#2}\Lambda_{#3}}}
\newcommand{\comp}[2]
{{(\theta_{#1},\varphi_{#1}), \Lambda_{#2},\{\bar V_{\eta,\Lambda_{#2}}\}_\eta}}
\newcommand{\ltj}[2]{\Lambda_{#1}+{#2} \bm e_{\Lambda_{#1}} }
\newcommand{\ltjp}[2]{{\Lambda'}_{#1}+{#2} \bm e_{\Lambda_{#1}} }
\newtheorem{nota}{Notation}[section]
\def\qed{{\unskip\nobreak\hfil\penalty50
\hskip2em\hbox{}\nobreak\hfil$\square$
\parfillskip=0pt \finalhyphendemerits=0\par}\medskip}
\def\proof{\trivlist \item[\hskip \labelsep{\bf Proof.\ }]}
\def\endproof{\null\hfill\qed\endtrivlist\noindent}
\def\proofof[#1]{\trivlist \item[\hskip \labelsep{\bf Proof of #1.\ }]}
\def\endproofof{\null\hfill\qed\endtrivlist\noindent}
\newcommand{\wrl}[2]{Y_{#1}^{\Lambda_0^{(#2)}}}
\newcommand{\wrlt}[2]{\tilde Y_{#1}^{\Lambda_0^{(#2)}}}
\newcommand{\ZZ}{\bbZ_2\times\bbZ_2}
\newcommand{\SSS}{\mathcal{S}}
\newcommand{\cs}{S}
\newcommand{\ct}{t}
\newcommand{\hS}{S}
\newcommand{\vv}{{\boldsymbol v}}
\newcommand{\ala}{a}
\newcommand{\bet}{b}
\newcommand{\gam}{c}
\newcommand{\alphas}{\alpha}
\newcommand{\alphai}{\alpha^{(\sigma_{1})}}
\newcommand{\alphan}{\alpha^{(\sigma_{2})}}
\newcommand{\betas}{\beta}
\newcommand{\betai}{\beta^{(\sigma_{1})}}
\newcommand{\betan}{\beta^{(\sigma_{2})}}
\newcommand{\alphass}{\alpha^{{(\sigma)}}}
\newcommand{\uu}{V}
\newcommand{\vp}{\varsigma}
\newcommand{\vpr}{R}
\newcommand{\tg}{\tau_{\Gamma}}
\newcommand{\sgg}{\Sigma_{\Gamma}}
\newcommand{\nh}{t28}
\newcommand{\rk}{6}
\newcommand{\nii}{2}
\newcommand{\nhh}{28}
\newcommand{\sjt}{30}
\newcommand{\sjtg}{30}
\newcommand{\bcg}{\caB(\caH_{\alpha})\otimes  C^{*}(\Sigma)}
\newcommand{\pza}[1]{\pi_0\lmk\caA_{\Lambda_{#1}}\rmk''}
\newcommand{\pzac}[1]{\pi_0\lmk\caA_{\Lambda_{#1}}\rmk'}
\newcommand{\pzacc}[1]{\pi_0\lmk\caA_{\Lambda_{#1}^c}\rmk'}
\newcommand{\trlzi}[2]{T_{#1}^{(\theta,\varphi) \Lambda_0^{(#2)}\unit}}

\newcommand{\obk}{\omega_{\mathop{\mathrm{bk}}}}
\newcommand{\obd}{\omega_{\mathop{\mathrm{bd}}}}
\newcommand{\obdm}{\omega_{\mathop{\mathrm{bd}(-)}}}
\newcommand{\abk}{\caA_{\mathbb Z^2}}
\newcommand{\hu}{\mathop {\mathrm H_{U}}}
\newcommand{\hd}{\mathop {\mathrm H_{D}}}

\newcommand{\abd}{\caA_{\hu}}
\newcommand{\aloch}{\caA_{\hu,\mathop{\mathrm {loc}}}}
\newcommand{\alocg}[1]{\caA_{#1,\mathop{\mathrm {loc}}}}
\newcommand{\hbk}{\caH_{\mathop{\mathrm{bk}}}}
\newcommand{\hbd}{\caH_{\mathop{\mathrm{bd}}}}
\newcommand{\hbdm}{\caH_{\mathop{\mathrm{bd}(-)}}}
\newcommand{\pbk}{\pi_{\mathop{\mathrm{bk}}}}
\newcommand{\pbd}{\pi_{\mathop{\mathrm{bd}}}}
\newcommand{\pbdm}{\pi_{\mathop{\mathrm{bd}(-)}}}
\newcommand{\mopbk}{{\mathop{\mathrm{bk}}}}
\newcommand{\mopbd}{{\mathop{\mathrm{bd}}}}
\newcommand{\Obk}{O_{\mathop{\mathrm{bk}}}}
\newcommand{\OUbk}{O_{\mathop{\mathrm{bk}}}^U}
\newcommand{\Orbd}{O^r_{\mathop{\mathrm{bd}}}}
\newcommand{\OrUbd}{O^{rU}_{\mathop{\mathrm{bd}}}}
\newcommand{\Obkl}{O_{\mathop{\mathrm{bk},\Lambda_0}}}
\newcommand{\OUbkl}{O_{\mathop{\mathrm{bk},\Lambda_0}}^U}
\newcommand{\Orbdl}{O^r_{\mathop{\mathrm{bd},\Lambda_{r0}}}}
\newcommand{\OrUbdl}{O^{rU}_{\mathop{\mathrm{bd},\Lambda_{r0}}}}
\newcommand{\Obu}{O_{\mathop{\mathrm{bd}}}^{\mathop{\mathrm{bu}}}}
\newcommand{\Obul}{O_{\mathop{\mathrm{bd}},\lz}^{\mathop{\mathrm{bu}}}}
\newcommand{\Odl}{O_{\lzr}^{\caD}}
\newcommand{\fd}{F^{\caD}}
\newcommand{\hilb}{\mathop{\mathrm {Hilb}}_f}
\newcommand{\Obun}[1]{O_{\mathop{\mathrm{bd}#1}}^{\mathop{\mathrm{bu}}}}
\newcommand{\OrUbdn}[1]{O^{rU}_{\mathop{\mathrm{bd}#1 }}}
\newcommand{\OUbkm}{O_{\mathop{\mathrm{bk}}}^{U(-)}}
\newcommand{\Orbdm}{O^{r(-)_{\mathop{\mathrm{bd} }}}}
\newcommand{\OrUbdm}{O^{rU(-)}_{\mathop{\mathrm{bd}}}}
\newcommand{\OUbklm}{O_{\mathop{\mathrm{bk},\Lambda_{0(-)}}}^{U(-)}}
\newcommand{\Orbdlm}{O^{r(-)}_{\mathop{\mathrm{bd},\Lambda_{r0(-)}}}}
\newcommand{\OrUbdlm}{O^{rU(-)}_{\mathop{\mathrm{bd},\Lambda_{r0(-)}}}}
\newcommand{\Obum}{O_{\mathop{\mathrm{bd}}}^{\mathop{\mathrm{bu}(-)}}}
\newcommand{\Obulm}{O_{\mathop{\mathrm{bd}},{\lm {0(-)}}}^{\mathop{\mathrm{bu}(-)}}}
\newcommand{\Otot}{O_{\mathop{\mathrm{total}}}}
\newcommand{\Ototl}{O_{\mathop{\mathrm{total}}}^{\lz,\lzm}}

\newcommand{\Obj}{{\mathop{\mathrm{Obj}}}}
\newcommand{\Mor}{{\mathop{\mathrm{Mor}}}}
\newcommand{\End}{{\mathop{\mathrm{End}}}}
\newcommand{\ti}[4]{\tilde\iota\lmk(#1,#2), (#3,#4)\rmk}

\newcommand{\Cabkl}{C_{\mathop{\mathrm{bk},\lz}}}
\newcommand{\Cabk}{C_{\mathop{\mathrm{bk}}}}
\newcommand{\CaUbk}{C_{\mathop{\mathrm{bk}}}^U}
\newcommand{\Carbd}{C^r_{\mathop{\mathrm{bd}}}}
\newcommand{\CarUbd}{C^{rU}_{\mathop{\mathrm{bd}}}}
\newcommand{\CaUbkl}{C_{\mathop{\mathrm{bk},\Lambda_0}}^U}
\newcommand{\Carbdl}{C^r_{\mathop{\mathrm{bd},\Lambda_{r0}}}}
\newcommand{\CarUbdl}{C^{rU}_{\mathop{\mathrm{bd},\Lambda_{r0}}}}
\newcommand{\Cabu}{C_{\mathop{\mathrm{bd}}}^{\mathop{\mathrm{bu}}}}
\newcommand{\Cabul}{C_{\mathop{\mathrm{bd}},\lz}^{\mathop{\mathrm{bu}}}}
\newcommand{\Cadl}{C_{\lzr}^{\caD}}
\newcommand{\Hom}{\mathop{\mathrm{Hom}}}
\newcommand{\Catotl}{C_{\mathop{\mathrm{total}}}^{\lz,\lzm}}
\newcommand{\Cafin}{C_{\mathop{\mathrm{bd}},\lz}^{\mathop{\mathrm{fin}}}}

\newcommand{\mm}{\Mor_{\tilde\caM}}
\newcommand{\om}{\Obj\lmk {\tilde\caM}\rmk}
\newcommand{\omt}{\otimes_{\tilde\caM}}

\newcommand{\CaUbkm}{C_{\mathop{\mathrm{bk}}}^{U(-)}}
\newcommand{\Carbdm}{C^{r(-)}_{\mathop{\mathrm{bd}}}}
\newcommand{\CarUbdm}{C^{rU(m)}_{\mathop{\mathrm{bd}}}}
\newcommand{\CaUbklm}{C_{\mathop{\mathrm{bk},\Lambda_{0(-)}}}^{U(-)}}
\newcommand{\Carbdlm}{C^{r(-)}_{\mathop{\mathrm{bd},\Lambda_{r0(-)}}}}
\newcommand{\CarUbdlm}{C^{rU(-)}_{\mathop{\mathrm{bd},\Lambda_{r0(-)}}}}
\newcommand{\Cabum}{C_{\mathop{\mathrm{bd}}}^{\mathop{\mathrm{bu}(-)}}}
\newcommand{\Cabuln}{C_{\mathop{\mathrm{bd}},\lm{0(-)}}^{\mathop{\mathrm{bu}(-)}}}
\newcommand{\Irr}{\mathop{\mathrm{Irr}}}


\newcommand{\Tbk}[3]{T_{#1}^{(\frac{3\pi}2,\frac\pi 2), {#2},{#3}}}
\newcommand{\Vrl}[2]{V_{#1,#2}}
\newcommand{\Tbkv}[2]{T_{#1}^{(\frac{3\pi}2,\frac\pi 2), #2, V_{#1,#2}}}
\newcommand{\Tbd}[3]{T_{#1}^{\mathrm{(l)} #2 #3}}
\newcommand{\Tbdv}[2]{T_{#1}^{\mathrm{(l)} #2\Vrl{{#1}}{#2}}}

\newcommand{\zc}{\caZ\lmk\Cabul\rmk}
\newcommand{\hi}[2]{\hat\iota\lmk #1: #2\rmk}

\newcommand{\lz}{\Lambda_0}
\newcommand{\lzr}{\Lambda_{r0}}
\newcommand{\lm}[1]{{\Lambda_{#1}}}
\newcommand{\llz}{(\lz,\lzr)}
\newcommand{\lmr}[1]{{\Lambda_{r#1}}}
\newcommand{\hlm}[1]{{\hat\Lambda_{#1}}}
\newcommand{\tlm}[1]{{\tilde\Lambda_{#1}}}
\newcommand{\ld}{\Lambda}
\newcommand{\pc}{\mathcal{PC}}
\newcommand{\Cbk}{\caC_{\mopbk}}
\newcommand{\CUbk}{\caC_{\mopbk}^U}
\newcommand{\Crbd}{\caC_{\mopbd}^r}
\newcommand{\Clbd}{\caC_{\mopbd}^l}
\newcommand{\Vbk}[2]{\caV_{#1#2}^{\mopbk}}
\newcommand{\Vbd}[2]{\caV_{#1#2}^{\mopbd}}
\newcommand{\VUbd}[2]{\caV_{#1#2}^{\mopbd U} }
\newcommand{\Vbu}[2]{\caV_{#1#2}^{\mathop{\mathrm{bu}}}}
\newcommand{\bl}{\caB_l}
\newcommand{\fbk}{\caF_{\mopbk}^U}
\newcommand{\fbd}{\caF_{\mopbd}^U}
\newcommand{\gu}{\caG^U}
\newcommand{\hb}[2]{\iota^{(\lz,\lzr)}\lmk#1: #2\rmk}
\newcommand{\hfc}{\hat F^{\llz}}
\newcommand{\ffc}{F^{\llz}}

\newcommand{\lzm}{\Lambda_{0(-)}}
\newcommand{\lzrm}{\Lambda_{r0(-)}}
\newcommand{\pcm}{\mathcal{PC(-)}}
\newcommand{\CUbkm}{\caC_{\mopbk}^{U(-)}}
\newcommand{\Crbdm}{\caC_{\mopbd}^{r(-)}}
\newcommand{\Clbdm}{\caC_{\mopbd}^{l(-)}}
\newcommand{\Vbkm}[2]{\caV_{#1#2}^{\mopbk (-)}}
\newcommand{\Vbdm}[2]{\caV_{#1#2}^{\mopbd(-)}}
\newcommand{\VUbdm}[2]{\caV_{#1#2}^{\mopbd U(-)} }
\newcommand{\Vbum}[2]{\caV_{#1#2}^{\mathop{\mathrm{bu}(-)}}}
\newcommand{\fbkm}{\caF_{\mopbk}^{U(-)}}
\newcommand{\fbdm}{\caF_{\mopbd}^{U(-)}}
\newcommand{\vrdh}[2]{\caV_{#1\ld_{#2}}}
\newcommand{\vrd}{\caV_{\rho\ld}}
\newcommand{\Vrd}{V_{\rho\ld}}
\newcommand{\tvrd}{\tilde\caV_{\rho\ld}}

\newcommand{\lr}[2]{\Lambda_{(#1,0),#2,#2}}
\newcommand{\lef}[2]{\overline{\Lambda_{(#1,0),\pi-#2,#2}}}
\newcommand{\lrhu}{(\Lambda_r)^c\cap \hu}
\newcommand{\lhuc}{\Lambda^c\cap \hu}
\newcommand{\lhu}{\Lambda\cap \hu}
\newcommand{\lrhuz}{(\Lambda_{r0})^c\cap \hu}
\newcommand{\lhucz}{\Lambda_{0}^c\cap \hu}
\newcommand{\lc}[1]{(\lm #1)^c\cap\hu}

\newcommand{\zam}{\caZ_a\lmk\tilde\caM\rmk }
\newcommand{\ozam}{\Obj\lmk \zam\rmk}
\newcommand{\mzam}{\Mor_{\zam}}
\newcommand{\rpc}[2]{\lmk(#1,#2), C_{(#1,#2)}\rmk}
\newcommand{\ozt}{\otimes_{\zam}}
\newcommand{\ool}{\caO_{\omega,\ld_0}}
\newcommand{\ools}[1]{\caO_{#1,\ld_0}}
\newcommand{\col}{C_{\omega,\ld_0}}
\newcommand{\ctvz}{\caC_{(\theta_0,\varphi_0)}}
\newcommand{\cols}[1]{C_{#1,\ld_0}}
\newcommand{\coll}[1]{C_{\omega,#1}}
\newcommand{\tool}{\tilde \caO_{\omega,\ld_0}}
\newcommand{\tcols}[1]{\tilde C_{#1,\ld_0}}
\newcommand{\tcol}{\tilde C_{\omega,\ld_0}}
\newcommand{\folz}{F_{\omega,\varphi,\ld_0}}
\newcommand{\fol}[2]{F_{#1,#2,\ld_0}}
\newcommand{\change}[1]{{#1}}
\title{A note on invariants of mixed-state topological order in 2D}
\author{Yoshiko Ogata \\Research Institute for Mathematical Sciences\\
 Kyoto University, Kyoto 606-8502 JAPAN}
\maketitle
\begin{abstract}
The classification of mixed-state topological order requires indices that behave monotonically under
finite-depth quantum channels.
In two dimensions, a braided $C^*$-tensor category, which corresponds to strong symmetry, arises from a state satisfying approximate Haag duality.
In this note, we show that 
the $S$-matrix and topological twists of the braided $C^*$-tensor category are quantities that are monotone under finite-depth quantum channels.

\end{abstract}
\section{Introduction}
The classification of topological phases of matter has attracted significant attention over the last two decades
\cite{lu2020detecting}
\cite{Coser2019MixedStatePhases}
\cite{Wang2025IntrinsicMixedState}
\cite{Sohal2025NoisyApproach}
\cite{Ellison2025ClassificationMixedState}
\cite{Fan2024DiagnosticsMixedState}
\cite{Bao2025ErrorFieldDecoherence}
\cite{Li2024ReplicaTopologicalOrder}
\cite{sang2025stability}
\cite{Zhang2024OneFormSymmetry}\cite{MTO}.
In gapped ground states, hence pure states in $2D$, many results are now known \cite{wen2004quantum}\cite{kitaev2006anyons}\cite{wen2016theory}\cite{wen2017colloquium}
\cite{dennis2002topological}\\
\cite{konig2014generating}\cite{MTC}.
By the adiabatic theorem \cite{hastings2005quasiadiabatic}\cite{bachmann2012automorphic}\cite{moon2020automorphic}, the classification of gapped ground states
boils down to the following question: is there a quasi-local automorphism
connecting two given gapped ground states?
In two dimensions, a braided $C^*$-tensor category was derived for
systems satisfying approximate Haag duality\cite{MTC}, and this category has been shown that it is actually an invariant of the classification. 
Namely, the braided $C^*$-tensor categories of $\omega$ and $\omega\alpha$
are equivalent if $\alpha$ is a quasi-local automorphism.
This braided $C^*$-tensor category physically can be understood as anyons.
In fact, in known exactly solvable models, it has been shown that this braided $C^*$-tensor category coincides with what are called anyons in those models
\cite{N2}\cite{bols2026category}.
The Haag duality is known to hold in all such exactly solvable models \cite{N1}\cite{FN}\cite{ogata2025haag}.

For the pure state classification, it is natural to consider quasi-local automorphisms or finite-depth quantum circuit as a basic operation.
For mixed state classification, the basic operation should be finite-depth quantum channels.
And what we would expect to have are some indices which {\it decrease} under finite-depth quantum channels.
Namely, what we should find there is not an {\it equivalence relation}, but
a {\it preorder relation}.

In \cite{MTO}, we derived a braided $C^*$-tensor category 
for mixed states satisfying approximate Haag duality.
This braided $C^*$-tensor category, in the physics terminology, corresponds to the
anyons related to the strong symmetry.
We showed that the category $\caD_{\omega\Phi}$ associated to the state $\omega\Phi$
is a subcategory of the category $\caD_\omega$ associated to the state $\omega$,
if $\Phi$ is a finite-depth quantum circuit.
This seemingly shows that the braided $C^*$-tensor category is what we are looking for,
namely something that {\it decreases}.
But there is a very confusing example \cite{Ellison2025ClassificationMixedState}
\cite{sang2025mixed}\cite{LO}.
The decohered Toric code model under $X$-channels can be obtained out of
a product state. However, the corresponding braided $C^*$-tensor category is non-trivial.
The point is that the category $\caD_{\omega\Phi}$ is not a {\it full} category of $\caD_\omega$.
This example makes us wonder if the strong symmetry-based category
obtained in \cite{MTO} would tell us anything about mixed state topological order.

In this note, we claim that they do tell us something.
We show that the $S$-matrix and twists of the braided $C^*$-tensor category
are some {\it decreasing indices}.
Namely, there is some preorder relation between the $S$-matrix and twists 
of $\omega$ and that of $\omega\Phi$.
It was pointed out in \cite{Wang2025IntrinsicMixedState}
that existence of abelian anyons with nontrivial twist or braiding means the mixed state cannot be
constructed out of product states with finite-depth quantum channels.
Our result fits with this statement, and can be understood as a generalization of this statement.

Let us recall the results in \cite{MTO}.
We introduced the mixed state version of approximate Haag duality.
The precise definition of the approximate Haag duality  is given in Appendix \ref{mah}.
Here, just to reduce the burden on the reader, we just recall the genuine Haag duality, a stronger condition. 
\begin{defn}
Let $\caA$ be a $2$-dimensional quantum spin system.
Let $\omega$ be a state on $\caA$
with a GNS representation $(\caH,\pi)$.
We say $\omega$ satisfies the Haag duality, if
for any cone $\Lambda$ 
\[
\pi(\caA_{\Lambda^c})'\cap \pi\lmk \caA\rmk''=
\pi\lmk\caA_{\Lambda}\rmk''
\]
holds.
\end{defn}
By Theorem 1.10, Theorem 1.11,  Lemma 4.4 of \cite{MTO}, to each state $\omega$ satisfying (approximate) Haag duality, we can associate
a strict braided $C^*$-tensor category $\caD_\omega$, defined up to equivalence of braided $C^*$-tensor categories.
We denote by $\Delta^{\caD_\omega}$ the set of all isomorphic classes of 
simple objects with a conjugate in $\caD_{\omega}$.
For each $a,b\in \Delta^{\caD_{\omega}}$, set 
\begin{align}
\begin{split}
&\theta_{a}^\omega:= \Tr_{\tau_a\otimes\tau_a}\lmk\epsilon(\tau_a,\tau_a)\rmk,\\
&S_{a,b}^\omega:= \Tr_{\tau_a\otimes\tau_b}\lmk
\epsilon(\tau_b,\tau_a)\epsilon(\tau_a,\tau_b)\rmk.
\end{split}
\end{align}
with the braiding $\epsilon$ of  $\caD_\omega$.
Here, $\tau_a\in\Obj\caD_\omega$ is a representative of $a$.
These values do not depend on the choice of the representatives $\tau_a$, $\tau_b$.
See section \ref{sts}.
We also denote by $d^\omega(a)$ the dimension of $a\in \Delta^{\caD_\omega}$.
We introduce a $|\Delta^{\caD_\omega}|\times |\Delta^{\caD_\omega}|$-matrix
(the size $|\Delta^{\caD_\omega}|$ of $\Delta^{\caD_\omega}$ may be infinite)
\begin{align}
\begin{split}
S^\omega:=\lmk S_{a,b}^\omega\rmk_{(a,b)\in \Delta^{\caD_\omega}\times\Delta^{\caD_\omega}},
\end{split}
\end{align}
and size $|\Delta^{\caD_\omega}|$ vectors
\begin{align}
\begin{split}
\theta^\omega:=\lmk\theta_a^\omega\rmk_{a\in\Delta^{\caD_\omega}},\quad
d^{\omega}:= \lmk d^\omega(a)\rmk_{a\in\Delta^{\caD_\omega}}
\end{split}
\end{align}

Theorem 1.14 of \cite{MTO} applied to finite-depth quantum channels is the following.
\begin{thm}\label{MTO}
Let $\caA$ be a $2$-dimensional quantum spin system.
Let $\omega_1$, $\omega_2$ be 
states on  $\caA$ satisfying the approximate Haag duality.
Let $\Phi$ be a finite-depth quantum channel.
Suppose that $\omega_2=\omega_1\circ\Phi$.
Then there is a faithful unitary braided tensor functor from
$\caD_{\omega_2}$ to $\caD_{\omega_1}$.
\end{thm}
\begin{proof}
Applying the Steinspring dilation theorem to each local channel consisting
$\Phi$, we see that there are a $2$-dimensional quantum spin system $\caB$,
a pure infinite product state $\psi$ on $\caB$,
 and
 an approximately factorizable automorphism $\alpha$ such that
 $\Phi=(\id_\caA\otimes\psi)\alpha\vert_{\caA}$.
 Hence we may apply Theorem 1.14 of \cite{MTO}.
\end{proof}
The main result of this paper is the following.
\begin{thm}\label{main}
Let $\caA$ be a $2$-dimensional quantum spin system.
Let $\omega_1$, $\omega_2$ be 
states on  $\caA$ satisfying the approximate Haag duality.
Suppose that there exists a finite-depth quantum channel $\Phi$
such that $\omega_2=\omega_1\circ\Phi$.
Then the following hold.
\begin{description}
\item[(1)]
There exist a $|\Delta^{\caD_{\omega_2}}|\times |\Delta^{\caD_{\omega_1}}|$ matrix $X$,
$|\Delta^{\caD_{\omega_1}}|\times |\Delta^{\caD_{\omega_2}}|$ matrix $Y$,
such that
\begin{align}\label{xsy}
\begin{split}
S^{\omega_2}=X S^{\omega_1}Y,
\end{split}
\end{align}
and
\begin{align}\label{xdyd}
\begin{split}
Xd^{\omega_1}=d^{\omega_2},\quad
Y^Td^{\omega_1}=d^{\omega_2}.
\end{split}
\end{align}
Here $Y^T$ denotes transpose of the matrix $Y$.
The matrices $X$ and $Y$ have non-negative entries.
Each row of $X$ has at most a finite number of non-zero elements, and
each column of $Y$ has at most a finite number of non-zero elements.
As a result, the matrix multiplications in (\ref{xsy}) and (\ref{xdyd}) are well defined.
\item[(2)]
There exist a $|\Delta^{\caD_{\omega_2}}|\times |\Delta^{\caD_{\omega_1}}|$ matrix
$N$ such that
\begin{align}\label{ns}
\begin{split}
\theta^{\omega_2}=N\theta^{\omega_1}.
\end{split}
\end{align}
The matrix $N$ has  non-negative integer entries.
Each row of $N$ has at most a finite number of non-zero elements.
As a result, the matrix multiplication in (\ref{ns}) is well defined.
\end{description}
\end{thm}
Let $\mathfrak S$ be the set of all (possibly infinite) square matrices and 
$\mathfrak V$ be the set of all (possibly infinite) vectors.
Let $\mathfrak R$ be the set of all (not necessarily square) matrices whose rows have at most a
finite number of non-zero
entries, and $\mathfrak R_{+}$ be the set of all elements in $\mathfrak R$
with non-negative entries, while $\mathfrak R_{\bbN}$ be the set of all elements in $\mathfrak R$
with non-negative integer entries.
Note that $\mathfrak R_{+}$ is closed under the multiplication that is defined when the
dimensions of the matrices fit.
Similarly, $\mathfrak R_{\bbN}$ is closed under the multiplication.

For each $S\in \mathfrak S$, we denote by $\nu_S$ the size of the matrix $S$.
On $\mathfrak S$, we introduce the relation $\prec_s$
by
$S_2\prec_s S_1$ if there exist 
$\nu_{S_2}\times \nu_{S_1}$ matrices $X, Y\in \mathfrak R_{+}$
such that 
$S_2=X S_1 Y^T$.
On $\mathfrak V$,  we introduce the relation $\prec_t$
by 
$\theta_2\prec_t\theta_1$ if there exists a  matrix
$N\in \mathfrak R_{\bbN}$ such that
$\theta_2=N\theta_1$.
Because $\mathfrak R_{+}$ and $\mathfrak R_{\bbN}$ are closed under multiplication,
both $\prec_s$ and $\prec_t$ are preorders.
Theorem \ref{main} says that, with respect to this preorder, we have
$S^{\omega_2}\prec_s S^{\omega_1}$
and $\theta^{\omega_2}\prec_t \theta^{\omega_1}$.

\section{$S$-matrix and twist}\label{sts}
First we recall basic facts from \cite{NT}.
Let $\caD$ be a strict braided $C^*$-tensor category. We occasionally denote by $[\rho]$
the isomorphic class of an object $\rho$ in $\caD$.
An object $\bar\rho\in\Obj \caD$ is said to be conjugate to $\rho\in\Obj\caD$
if there exist morphisms $R\in\Mor_{\caD}\lmk\unit, \bar\rho\otimes\rho\rmk$
and $\bar R\in\Mor_{\caD}\lmk\unit, \rho\otimes\bar\rho\rmk$
such that
\begin{align}\label{ce}
\begin{split}
&\lmk\bar R^*\otimes\id_\rho\rmk\lmk \id_\rho\otimes R\rmk=\id_\rho,
\quad \lmk R^*\otimes\id_{\bar\rho}\rmk\lmk \id_{\bar\rho}\otimes \bar R\rmk
=\id_{\bar\rho}.
\end{split}
\end{align}
In this case, we say $(R,\bar R)$ is a solution of the conjugate equations for $(\rho,\bar\rho)$.
We denote by $\Delta^\caD$ the set of all isomorphic classes of 
simple objects  with a conjugate in $\caD$.
We fix a representative ${\tau_c}$ for each $c\in\Delta^{\caD}$.

If $\rho\in\Obj\caD$ has a conjugate, its endomorphism space $\End_{\caD}(\rho)$ is finite dimensional
(Proposition 2.2.8 \cite{NT}). 
Therefore, $\rho$ can be written as a direct sum of simple objects. Namely,
there are a finite subset $\Delta_\rho\subset \Delta^{\caD}$ and 
 isometries $u_c^\rho(\mu)\in\Mor_{\caD}(\tau_c,\rho)$, $c\in\Delta_\rho$, $\mu=1,\ldots,n_{\rho,c}$ with
multiplicity $n_{\rho,c}\in \bbN$,
such that
\begin{align}\label{decom}
\begin{split}
\id_{\rho}=\sum_{c\in\Delta_\rho}\sum_{\mu=1}^{n_{\rho,c}}
u_c^\rho(\mu)u_c^\rho(\mu)^*.
\end{split}
\end{align}
Recall that the class of objects in $\caD$ that have conjugates forms a $C^*$-tensor subcategory of $\caD$
(See Proposition 2.2.10\cite{NT}).
A solution $(R,\bar R)$ of the conjugate equations for $(\rho,\bar\rho)$
is called standard if it is of the form
\begin{align}
\begin{split}
R=\sum_{c\in\Delta_\rho}\sum_{\mu=1}^{n_{\rho,c}}\lmk\bar u_c^\rho(\mu)\otimes u_c^\rho(\mu)\rmk R_c,\\
\bar R=\sum_{c\in\Delta_\rho}\sum_{\mu=1}^{n_{\rho,c}}\lmk u_c^\rho(\mu)\otimes \bar u_c^\rho(\mu)\rmk \bar R_c,
\end{split}
\end{align}
with $(R_c, \bar R_c)$, a solution of the conjugate equations for $(\tau_c,\bar\tau_c)$
with $\lV R_c\rV=\lV \bar R_c\rV=d(c)$. Here, $d(c)$ denotes the dimension of $c$.
Here, $\bar u_c^\rho(\mu)\in\Mor_{\caD}(\bar \tau_c,\bar\rho)$ are isometries with 
$\sum_{c\in\Delta_\rho}\sum_{\mu=1}^{n_{\rho,c}} \bar u_c^\rho(\mu)\bar u_c^\rho(\mu)^*=\id_{\bar\rho}$.
A standard solution $(R,\bar R)$ defines a positive faithful trace on the finite dimensional $C^*$-algebra
$\End_{\caD}(\rho)$ via
\begin{align}
\begin{split}
\Tr_{\rho} (T)=R^*\lmk\id_{\bar\rho}\otimes T\rmk R=\bar R^*\lmk T\otimes \id_{\rho}\rmk \bar R,\quad
T\in \End_{\caD}(\rho).
\end{split}
\end{align}
This functional does not depend on the choice of the standard solution $(R,\bar R)$. (See Theorem 2.2.16 \cite{NT}.)
If
$(R_\rho,\bar R_\rho)$, $(R_\sigma,\bar R_\sigma)$ are standard solutions of the conjugate equations for
$(\rho,\bar\rho)$, $(\sigma,\bar\sigma)$ respectively,
then we have
\begin{align}\label{tracef}
\begin{split}
\Tr_{\rho\otimes\sigma}\lmk
T
\rmk
=\lmk R_\rho^*\otimes\bar R_\sigma^*\rmk
\lmk\id_{\bar\rho}\otimes T\otimes\id_{\bar\sigma}\rmk
\lmk  R_\rho \otimes\bar R_\sigma\rmk
\end{split}
\end{align}
(See proof of Theorem 2.2.18 \cite{NT}.)
%
\begin{defn}
Let $\epsilon$ be the braiding of  $\caD$.
For any isomorphic classes $[\rho],[\sigma]$ with conjugate in $\Obj \caD$, we set
\begin{align}
\begin{split}
&\theta_{[\rho]}:= \Tr_{\rho\otimes\rho}\lmk\epsilon(\rho,\rho)\rmk,\\
&S_{[\rho],[\sigma]}:= \Tr_{\rho\otimes\sigma}\lmk
\epsilon(\sigma,\rho)\epsilon(\rho,\sigma)\rmk.
\end{split}
\end{align}
These values do not depend on the choice of the representatives $\rho$, $\sigma$
due to the tracial property of $\Tr$ and the naturality of $\epsilon$.
Normalizations of them are called S-matrix and twist. 
\end{defn}
\begin{lem}\label{lem2}
Let $(R_\rho,\bar R_\rho)$, $(R_\sigma,\bar R_\sigma)$ be (not necessarily standard) solutions of the conjugate equations for
$(\rho,\bar\rho)$, $(\sigma,\bar\sigma)$ respectively.
Then there are positive numbers $t_a^\rho$, $s_b^\sigma$ for $a\in \Delta_\rho$,  $b\in\Delta_\sigma$
such that 
\begin{align}
\begin{split}
&\lmk R_\rho^*\otimes \bar R_\sigma^*\rmk
\lmk\id_{\bar\rho}\otimes \epsilon(\sigma,\rho)\epsilon(\rho,\sigma)\otimes \id_{\bar\sigma}\rmk
\lmk R_\rho\otimes \bar R_\sigma\rmk=\sum_{a\in\Delta_\rho, b\in \Delta_{\sigma}} t_a^\rho s_b^\sigma S_{a,b}
\end{split}
\end{align}
and
\begin{align}
\sum_{a\in\Delta_\rho} d(a) t_a^\rho=\lV R_\rho\rV^2,\quad
\sum_{b\in\Delta_\sigma} d(b) s_b^\sigma=\lV \bar R_\sigma\rV^2.
\end{align}
\end{lem}
\begin{proof}
Let $(R'_\rho,\bar R'_\rho)$, $(R'_\sigma,\bar R'_\sigma)$ be standard solutions of the conjugate equations for
$(\rho,\bar\rho)$, $(\sigma,\bar\sigma)$ respectively.
By Proposition 2.2.5 \cite{NT} for $(\bar\rho,\rho)$, there exist an invertible
$T\in \End_{\caD}(\rho)$ and an invertible $W\in\End_{\caD}(\sigma)$
such that
\begin{align}
\begin{split}
&R_\rho=\lmk \id_{\bar\rho}\otimes T^*\rmk R'_\rho,\quad
\bar R_\rho=\lmk T^{-1} \otimes\id_{\bar\rho} \rmk \bar R'_\rho,\\
&R_\sigma=\lmk \id_{\bar\sigma}\otimes W^*\rmk R'_\sigma,\quad
\bar R_\sigma=\lmk W^{-1} \otimes\id_{\bar\sigma} \rmk \bar R'_\sigma.
\end{split}
\end{align}
 Substituting those, using (\ref{tracef}), we have
 \begin{align}
\begin{split}
&\lmk R_\rho^*\otimes \bar R_\sigma^*\rmk
\lmk\id_{\bar\rho}\otimes \epsilon(\sigma,\rho)\epsilon(\rho,\sigma)\otimes \id_{\bar\sigma}\rmk
\lmk R_\rho\otimes \bar R_\sigma\rmk\\
&=
\lmk {R'_\rho}^*\otimes \lmk{\bar R}'_\sigma\rmk^*\rmk
\lmk\id_{\bar\rho}\otimes 
\lmk T\otimes  {W^*}^{-1} \rmk
\epsilon(\sigma,\rho)\epsilon(\rho,\sigma)
\lmk T^*\otimes  W^{-1} \rmk
\otimes \id_{\bar\sigma}\rmk
\lmk R'_\rho\otimes \bar R'_\sigma\rmk\\
&=\Tr_{\rho\otimes\sigma}
\lmk
\lmk T\otimes  {W^*}^{-1} \rmk
\epsilon(\sigma,\rho)\epsilon(\rho,\sigma)
\lmk T^*\otimes  W^{-1} \rmk
\rmk\\
&=\Tr_{\rho\otimes\sigma}
\lmk
\lmk T^*T\otimes  ({W^* W})^{-1} \rmk
\epsilon(\sigma,\rho)\epsilon(\rho,\sigma)
\rmk
\end{split}
\end{align}
Note from the naturality of $\epsilon$ and the decomposition (\ref{decom}),
we have
\begin{align}
\begin{split}
\epsilon(\sigma,\rho)\epsilon(\rho,\sigma)
=\sum_{a\in \Delta_\rho}\sum_{b\in \Delta_\sigma}
\sum_{\mu=1}^{n_{\rho,a}}\sum_{\nu=1}^{n_{\sigma,b}}
\lmk
u_a^\rho(\mu)\otimes u_b^\sigma(\nu)
\rmk
\epsilon\lmk \tau_b, \tau_a\rmk
\epsilon\lmk \tau_a, \tau_b\rmk
\lmk
u_a^\rho(\mu)\otimes u_b^\sigma(\nu)
\rmk^*.
\end{split}
\end{align}
Substituting this, we obtain
\begin{align}\label{mame}
\begin{split}
&\lmk R_\rho^*\otimes \bar R_\sigma^*\rmk
\lmk\id_{\bar\rho}\otimes \epsilon(\sigma,\rho)\epsilon(\rho,\sigma)\otimes \id_{\bar\sigma}\rmk
\lmk R_\rho\otimes \bar R_\sigma\rmk\\
&=\sum_{a\in \Delta_\rho}\sum_{b\in \Delta_\sigma}
\sum_{\mu=1}^{n_{\rho,a}}\sum_{\nu=1}^{n_{\sigma,b}}
\Tr_{\rho\otimes\sigma}
\lmk
\lmk T^*T\otimes  ({W^* W})^{-1} \rmk
\lmk \lmk
u_a^\rho(\mu)\otimes u_b^\sigma(\nu)
\rmk
\epsilon\lmk \tau_b, \tau_a\rmk
\epsilon\lmk \tau_a, \tau_b\rmk
\lmk
u_a^\rho(\mu)\otimes u_b^\sigma(\nu)
\rmk^*\rmk
\rmk\\
&=
\sum_{a\in \Delta_\rho}\sum_{b\in \Delta_\sigma}
\lmk \sum_{\mu=1}^{n_{\rho,a}} \lmk u_a^\rho(\mu)\rmk^*T^*T u_a^\rho(\mu)\rmk\cdot 
\lmk \sum_{\nu=1}^{n_{\sigma,b}}\lmk u_b^\sigma(\nu)\rmk^*({W^* W})^{-1}u_b^\sigma(\nu)\rmk
\Tr_{\tau_a\otimes\tau_b}\lmk \epsilon\lmk \tau_b, \tau_a\rmk
\epsilon\lmk \tau_a, \tau_b\rmk\rmk.
\end{split}
\end{align}
Here, we used the fact that $\lmk u_a^\rho(\mu)\rmk^* T^*T u_a^\rho(\mu)\in \End_{\caD}(\tau_a)$, $\lmk u_b^\sigma(\nu)\rmk^*({W^* W})^{-1}u_b^\sigma(\nu)\in \End_{\caD}(\tau_b)$
can be regarded as a scalar because of the irreducibility of $\tau_a$, $\tau_b$.
Note that
\begin{align}
\begin{split}
 &\sum_{\mu=1}^{n_{\rho,a}} \lmk u_a^\rho(\mu)\rmk^* T^*T u_a^\rho(\mu)\\
 &=\sum_{\mu=1}^{n_{\rho,a}} d(a)^{-1} \Tr_{\tau_a}\lmk \lmk u_a^\rho(\mu)\rmk^*T^*T u_a^\rho(\mu)\rmk\\
&=\sum_{\mu=1}^{n_{\rho,a}} d(a)^{-1}
\Tr_{\rho}\lmk u_a^\rho(\mu) \lmk u_a^\rho(\mu)\rmk^*T^*T \rmk \\
&=d(a)^{-1}\Tr_{\rho} p_a^\rho T^*T\\
&=:t_a^\rho
 \end{split}
\end{align}
Here, $p_a^\rho:=\sum_{\mu=1}^{n_{\rho,a}}   u_a^\rho(\mu)\lmk u_a^\rho(\mu)\rmk^*$ is the central projection
of $\End_{\caD}(\rho)$ corresponding to $a\in\Delta_\rho$.
Similarly, we have
\begin{align}
\begin{split}
&\sum_{\nu=1}^{n_{\sigma,b}}
\lmk u_b^\sigma(\nu)\rmk^*({W^* W})^{-1}u_b^\sigma(\nu)
=d(b)^{-1}\Tr_{\sigma}\lmk p_b^{\sigma}
({W^* W})^{-1}\rmk\\
&=:s_b^\sigma.
\end{split}
\end{align}
Because $\Tr$ is faithful, 
$t_a^\rho, s_b^\sigma>0$ and
\begin{align}
\begin{split}
&\lmk R_\rho^*\otimes \bar R_\sigma^*\rmk
\lmk\id_{\bar\rho}\otimes \epsilon(\sigma,\rho)\epsilon(\rho,\sigma)\otimes \id_{\bar\sigma}\rmk
\lmk R_\rho\otimes \bar R_\sigma\rmk
=\sum_{a\in\Delta_\rho, b\in \Delta_{\sigma}} t_a^\rho s_b^\sigma S_{a,b}.
\end{split}
\end{align}
From the definition, we have
\begin{align}
\begin{split}
&\sum_{a\in \Delta_\rho}d(a)t_a^\rho= \sum_{a\in \Delta_\rho}\Tr_{\rho} p_a^\rho T^*T=\Tr_{\rho}T^*T
=\lV R_\rho\rV^2,\\
&\sum_{b\in \Delta_\rho}d(b)s_b^\sigma=\sum_{b\in \Delta_\rho}\Tr_{\sigma}\lmk p_b^{\sigma}
({W^* W})^{-1}\rmk
=\Tr_{\sigma}\lmk 
({W^* W})^{-1}\rmk=\lV \bar R_\sigma\rV^2.
\end{split}
\end{align}
\end{proof}
\begin{lem}
Let $(R_\rho,\bar R_\rho)$ be a (not necessarily standard) solution of the conjugate equation for
$(\rho,\bar\rho)$.
Then 
\begin{align}
\begin{split}
&\lmk R_\rho^*\otimes \bar R_\rho^*\rmk
\lmk\id_{\bar\rho}\otimes \epsilon(\rho,\rho)\otimes \id_{\bar\rho}\rmk
\lmk R_\rho\otimes \bar R_\rho\rmk
=\sum_{a\in \Delta_\rho}n_{\rho,a}\theta(a).
\end{split}
\end{align}
\end{lem}
\begin{proof}
By the same argument as in the proof of Lemma \ref{lem2},
we have
\begin{align}
\begin{split}
&\lmk R_\rho^*\otimes \bar R_\rho^*\rmk
\lmk\id_{\bar\rho}\otimes \epsilon(\rho,\rho)\otimes \id_{\bar\rho}\rmk
\lmk R_\rho\otimes \bar R_\rho\rmk\\
&=\sum_{a,b\in \Delta_\rho}\sum_{\mu=1}^{n_{\rho,a}}\sum_{\nu=1}^{n_{\rho,b}}
\Tr_{\rho\otimes \rho}\lmk 
\lmk T^*T\otimes (T^*T)^{-1}\rmk
\lmk u_b^\rho(\nu)\otimes u_a^\rho(\mu)\rmk
\epsilon\lmk \tau_a,\tau_b\rmk
\lmk u_a^\rho(\mu)\otimes u_b^\rho(\nu)\rmk^*
\rmk\\
&=\sum_{a,b\in \Delta_\rho}\sum_{\mu=1}^{n_{\rho,a}}\sum_{\nu=1}^{n_{\rho,b}}
\Tr_{\tau_b\otimes \tau_a}\lmk 
\epsilon\lmk \tau_a,\tau_b\rmk
\lmk u_a^\rho(\mu)\otimes u_b^\rho(\nu)\rmk^*
\lmk T^*T\otimes (T^*T)^{-1}\rmk
\lmk u_b^\rho(\nu)\otimes u_a^\rho(\mu)\rmk
\rmk\\
&=\sum_{a,b\in \Delta_\rho}
\Tr_{\tau_b\otimes \tau_a}\lmk 
\epsilon\lmk \tau_a,\tau_b\rmk
\sum_{\mu=1}^{n_{\rho,a}}\sum_{\nu=1}^{n_{\rho,b}}
\lmk u_a^\rho(\mu)^*T^*T u_b^\rho(\nu)\otimes u_b^\rho(\nu)^* (T^*T)^{-1}u_a^\rho(\mu)\rmk
\rmk.
\end{split}
\end{align}
Because $ u_a^\rho(\mu)^*T^*T u_b^\rho(\nu)\in \Mor_{\caD}(\tau_b,\tau_a)$,
only the terms $a=b$ can be nonzero.
Therefore,
\begin{align}
\begin{split}
&\lmk R_\rho^*\otimes \bar R_\rho^*\rmk
\lmk\id_{\bar\rho}\otimes \epsilon(\rho,\rho)\otimes \id_{\bar\rho}\rmk
\lmk R_\rho\otimes \bar R_\rho\rmk\\
&=\sum_{a\in \Delta_\rho}
\Tr_{\tau_a\otimes \tau_a}\lmk 
\epsilon\lmk \tau_a,\tau_a\rmk
\sum_{\mu=1}^{n_{\rho,a}}\sum_{\nu=1}^{n_{\rho,a}}
\lmk u_a^\rho(\mu)^*T^*T u_a^\rho(\nu)\otimes u_a^\rho(\nu)^* (T^*T)^{-1}u_a^\rho(\mu)\rmk
\rmk\\
&=\sum_{a\in \Delta_\rho}
\sum_{\mu=1}^{n_{\rho,a}}\sum_{\nu=1}^{n_{\rho,a}}
\lmk u_a^\rho(\mu)^*T^*T u_a^\rho(\nu)u_a^\rho(\nu)^* (T^*T)^{-1}u_a^\rho(\mu)\rmk
\Tr_{\tau_a\otimes \tau_a}\lmk 
\epsilon\lmk \tau_a,\tau_a\rmk\rmk\\
&=\sum_{a\in \Delta_\rho}
\sum_{\mu=1}^{n_{\rho,a}}\lmk u_a^\rho(\mu)^*T^*T (T^*T)^{-1}u_a^\rho(\mu)\rmk
\Tr_{\tau_a\otimes \tau_a}\lmk 
\epsilon\lmk \tau_a,\tau_a\rmk\rmk\\
&=\sum_{a\in \Delta_\rho}n_{\rho,a}\theta(a)
\end{split}
\end{align}
\end{proof}
\section{Faithful braided tensor functor}
Let $\caD$, $\caD'$ be strict braided $C^*$-tensor categories and
$(F, F_0, F_2)$ be a faithful unitary braided tensor functor.
(See Definition 2.1.3 of \cite{NT}).
Namely, $F: \caD\to\caD'$ is a linear faithful functor preserving $*$-operation,
$F_0 : \unit_{\caD'}\to F(\unit_{\caD})$ a unitary isomorphism,
and unitary natural isomorphisms
\begin{align}
F_2(\rho,\sigma): F(\rho)\otimes F(\sigma)\to F(\rho\otimes\sigma),\quad \rho,\sigma\in \Obj \caD.
\end{align}
They satisfy
\begin{align}\label{kani}
\begin{split}
F_2(\rho\otimes\sigma, \gamma)\cdot
\lmk F_2(\rho,\sigma)\otimes \id_{F(\gamma)}\rmk
=F_2\lmk\rho,\sigma\otimes \gamma\rmk
\lmk \id_{F(\rho)}\otimes F_2(\sigma,\gamma)\rmk,
\end{split}
\end{align}
and
\begin{align}\label{ebi}
\begin{split}
&F_2(\unit, \rho)\lmk F_0\otimes \id_{F(\rho)}\rmk
=\id_{F(\rho)},\\
&F_2(\rho, \unit)\lmk \id_{F(\rho)} \otimes F_0\rmk
=\id_{F(\rho)},
\end{split}
\end{align}
for any $\rho,\sigma,\gamma\in\Obj\caD$.
\begin{lem}\label{rr}
If $(R_\rho,\bar R_\rho)$ is a solution of the conjugate equation for
$(\rho,\bar\rho)\in\Obj\caD^{\times 2}$,
then 
\begin{align}\label{ff}
({R'_\rho}, {\bar R'_\rho})
:=
\lmk F_2\lmk \bar\rho, \rho\rmk^* F(R_\rho)F_0, 
F_2(\rho,\bar\rho)^* F(\bar R_\rho) F_0\rmk
\end{align}
is a solution of conjugate equation
for $\lmk F(\rho), F(\bar\rho)\rmk$.
\end{lem}
\begin{proof}
We will omit the subscript $\rho$ from $R$s, in the proof.
From the definition, it is clear that $R'\in\Mor_{\caD'}(\unit_{\caD'}, F(\bar \rho)\otimes F(\rho))$
and $R\in \Mor_{\caD'}(\unit_{\caD'}, F( \rho)\otimes F(\bar \rho))$.
Using (\ref{kani}), (\ref{ebi}), the naturality of $F_2$, and the conjugate equation for $(R,\bar R)$,
we obtain the first conjugate equation for $(R',\bar R')$:
\begin{align}
\begin{split}
&\lmk\lmk \bar R'\rmk^*\otimes\id_{F(\rho)}\rmk\lmk \id_{F(\rho)}\otimes R'\rmk\\
&=
\lmk
F_0^*F(\bar R)^*\otimes \id_{F(\rho)}
\rmk
\lmk
F_2(\rho,\bar\rho) \otimes \id_{F(\rho)}
\rmk
\lmk
\id_{F(\rho)}\otimes F_2\lmk \bar\rho, \rho\rmk^*
\rmk
\lmk
\id_{F(\rho)}\otimes  F(R)F_0
\rmk\\
&=\lmk
F_0^*F(\bar R)^*\otimes \id_{F(\rho)}
\rmk
F_2(\rho\otimes\bar\rho, \rho)^*
F_2\lmk\rho,\bar\rho\otimes \rho\rmk
\lmk
\id_{F(\rho)}\otimes  F(R)F_0
\rmk\\
&=
\lmk F_2(\rho\otimes\bar\rho, \rho)
\lmk F(\bar R)F_0\otimes \id_{F(\rho)}\rmk\rmk^*
F_2\lmk\rho,\bar\rho\otimes \rho\rmk
\lmk
\id_{F(\rho)}\otimes F(R)F_0
\rmk\\
&=
\lmk
F\lmk \bar R\otimes \id_\rho\rmk
F_2(\unit, \rho)
F_0\otimes \id_{F(\rho)}\rmk^*
F\lmk\id_{\rho}\otimes R\rmk
F_2\lmk\rho,\unit\rmk
\lmk
\id_{F(\rho)}\otimes  F_0
\rmk \\
&=
\lmk
F_2(\unit, \rho)
F_0\otimes \id_{F(\rho)}\rmk^*
F\lmk \lmk \bar R^*\otimes \id_\rho\rmk\lmk\id_{\rho}\otimes R\rmk\rmk
F_2\lmk\rho,\unit\rmk
\lmk
\id_{F(\rho)}\otimes  F_0
\rmk \\
&=
\lmk
F_2(\unit, \rho)
F_0\otimes \id_{F(\rho)}\rmk^*
F_2\lmk\rho,\unit\rmk
\lmk
\id_{F(\rho)}\otimes  F_0
\rmk \\
&=\id_{F(\rho)}.
\end{split}
\end{align}
The second conjugate equation follows from the same argument.
\end{proof}
\begin{lem} \label{rtr}
For any  solutions $(R_\rho,\bar R_\rho)$, 
$(R_\sigma,\bar R_\sigma)$ of the conjugate equations for 
$(\rho,\bar\rho), (\sigma,\bar\sigma)\in \Obj\caD^{\times 2}$ respectively,
define $({R'_\rho}, {\bar R'_\rho})$ by the formula (\ref{ff}), and
$({R'_\sigma}, {\bar R'_\sigma})$ analogously.
Then for any $T\in \End_{\caD}\lmk\rho\otimes \sigma\rmk$, we have
\begin{align}
\begin{split}
&\lmk {R'_\rho}\otimes  {\bar R'_\sigma}\rmk^*
\lmk 
\id_{F(\bar \rho)}\otimes F_2(\rho,\sigma)^* F(T) F_2(\rho,\sigma)\otimes \id_{F(\bar\sigma)} 
\rmk
\lmk {R'_\rho}\otimes  {\bar R'_\sigma}\rmk\\
&=\lmk R_\rho^*\otimes\bar R_\sigma^*\rmk
\lmk\id_{\bar\rho}\otimes T\otimes\id_{\bar\sigma}\rmk
\lmk  R_\rho \otimes\bar R_\sigma\rmk.
\end{split}
\end{align}
\end{lem}
\begin{proof}
First, from the naturality of $F_2$,
\begin{align}
\begin{split}
&\id_{F(\bar \rho)}\otimes  F(T) \otimes \id_{F(\bar\sigma)} \\
&=\lmk \id_{F(\bar \rho)}\otimes  F(T) \rmk\otimes \id_{F(\bar\sigma)} \\
&=F_2\lmk\bar\rho, \rho\otimes \sigma \rmk^* 
F\lmk\id_{\bar\rho}\otimes T\rmk
F_2\lmk \bar\rho,\rho\otimes\sigma\rmk\otimes \id_{F(\bar\sigma)}\\
&=\lmk F_2\lmk\bar\rho, \rho\otimes \sigma \rmk\otimes \id_{F(\bar\sigma)}\rmk^* 
\lmk F\lmk\id_{\bar\rho}\otimes T\rmk\otimes \id_{F(\bar\sigma)}\rmk
\lmk F_2\lmk \bar\rho,\rho\otimes\sigma\rmk\otimes \id_{F(\bar\sigma)}\rmk\\
&=\lmk F_2\lmk\bar\rho, \rho\otimes \sigma \rmk\otimes \id_{F(\bar\sigma)}\rmk^* 
F_2\lmk\bar\rho\otimes\rho\otimes\sigma,\bar\sigma\rmk^*
F\lmk\id_{\bar\rho}\otimes T\otimes\id_{\bar\sigma}\rmk
F_2\lmk\bar\rho\otimes\rho\otimes\sigma,\bar\sigma\rmk
\lmk F_2\lmk \bar\rho,\rho\otimes\sigma\rmk\otimes \id_{F(\bar\sigma)}\rmk\\
\end{split}
\end{align}
Using (\ref{kani}), (\ref{ebi}) and naturality of $F_2$, we have
\begin{align}
\begin{split}
&F_2\lmk\bar\rho\otimes\rho\otimes\sigma,\bar\sigma\rmk
\lmk F_2\lmk \bar\rho,\rho\otimes\sigma\rmk\otimes \id_{F(\bar\sigma)}\rmk
\lmk 
\id_{F(\bar \rho)}\otimes F_2(\rho,\sigma)\otimes \id_{F(\bar\sigma)} 
\rmk
\lmk {R'_\rho}\otimes  {\bar R'_\sigma}\rmk\\
&=F_2\lmk\bar\rho\otimes\rho\otimes\sigma,\bar\sigma\rmk
\lmk \lmk F_2\lmk\bar\rho\otimes \rho,\sigma \rmk
\lmk F_2\lmk\bar\rho,\rho\rmk \otimes\id_{F(\sigma)} \rmk\rmk\otimes\id_{F(\bar\sigma)}\rmk
 \lmk F_2\lmk \bar\rho, \rho\rmk^* F(R_\rho)F_0\otimes
F_2(\sigma,\bar\sigma)^* F(\bar R_\sigma) F_0\rmk\\
&=F_2\lmk\bar\rho\otimes\rho\otimes\sigma,\bar\sigma\rmk
 \lmk F_2\lmk\bar\rho\otimes \rho,\sigma \rmk \otimes\id_{F(\bar\sigma)}\rmk
 \lmk F(R_\rho)F_0\otimes
F_2(\sigma,\bar\sigma)^* F(\bar R_\sigma) F_0\rmk\\
&=F_2\lmk\bar\rho\otimes\rho, \sigma\otimes \bar\sigma\rmk
 \lmk F(R_\rho)F_0\otimes
 F(\bar R_\sigma) F_0\rmk\\
 &=F(R_\rho\otimes \bar R_\sigma)
 F_2\lmk\unit, \unit\rmk
 \lmk F_0\otimes
  F_0\rmk\\
  &=F(R_\rho\otimes \bar R_\sigma)F_0
\end{split}
\end{align}
Substituting these, we obtain
\begin{align}
\begin{split}
&\lmk {R'_\rho}\otimes  {\bar R'_\sigma}\rmk^*
\lmk 
\id_{F(\bar \rho)}\otimes F_2(\rho,\sigma)^* F(T) F_2(\rho,\sigma)\otimes \id_{F(\bar\sigma)} 
\rmk
\lmk {R'_\rho}\otimes  {\bar R'_\sigma}\rmk\\
&=
\lmk F(R_\rho\otimes \bar R_\sigma)F_0\rmk^*
F\lmk\id_{\bar\rho}\otimes T\otimes\id_{\bar\sigma}\rmk
F(R_\rho\otimes \bar R_\sigma)F_0\\
&=F_0^* F\lmk 
(R_\rho\otimes \bar R_\sigma)^*\lmk\id_{\bar\rho}\otimes T\otimes\id_{\bar\sigma}\rmk
(R_\rho\otimes \bar R_\sigma)
\rmk
F_0\\
&=
(R_\rho\otimes \bar R_\sigma)^*\lmk\id_{\bar\rho}\otimes T\otimes\id_{\bar\sigma}\rmk
(R_\rho\otimes \bar R_\sigma).
\end{split}
\end{align}
\end{proof}
\section{Proof of Theorem \ref{main}}
Let $F$ be the faithful unitary braided tensor functor $F$ from
$\caD_{\omega_2}$ to $\caD_{\omega_1}$ given in Theorem \ref{MTO}.
Let $\zeta,\xi\in \Delta^{\caD_{\omega_2}}$
with representative $\tau_\zeta^{\omega_2}\in\zeta$, 
$\tau_\xi^{\omega_2}\in\xi$.
Then, from Lemma \ref{lem2}, Lemma \ref{rr} and Lemma \ref{rtr},
there are positive numbers $t_a^\zeta$, $s_b^\xi$ for $a\in \Delta_{F\lmk \tau_\zeta^{\omega_2}\rmk}$,  $b\in\Delta_{F\lmk \tau_\xi^{\omega_2}\rmk}$
such that 
\begin{align}
\begin{split}
&S_{\zeta,\xi}^{\omega_2}
=\sum_{a\in\Delta_{F\lmk \tau_\zeta^{\omega_2}\rmk}, b\in \Delta_{F\lmk \tau_\zeta^{\omega_2}\rmk}} t_a^\zeta s_b^\xi S_{a,b}^{\omega_1}
\end{split}
\end{align}
and
\begin{align}
\sum_{a\in\Delta_{F\lmk \tau_\zeta^{\omega_2}\rmk}} 
d^{\omega_1}(a) t_a^\zeta=d^{\omega_2}(\zeta),\quad
\sum_{b\in\Delta_{F\lmk \tau_\xi^{\omega_2}\rmk}} d^{\omega_1}(b) s_b^\xi=d^{\omega_2}(\xi).
\end{align}
Similarly, there are positive integers $n_{\zeta,a}$, $a\in \Delta_{F\lmk \tau_\zeta^{\omega_2}\rmk}$
\begin{align}
\begin{split}
\theta^{\omega_2}_\zeta=\sum_{a\in \Delta_{F\lmk \tau_\zeta^{\omega_2}\rmk}}n_{\zeta,a}\theta^{\omega_1}(a).
\end{split}
\end{align}
Define the $|\Delta^{\caD_{\omega_2}}|\times |\Delta^{\caD_{\omega_1}}|$-matrix
$X$ by $X_{\zeta,a}:= t_a^\zeta$ if $a\in\Delta_{F\lmk \tau_\zeta^{\omega_2}\rmk}$,
and $X_{\zeta,a}:= 0$ if $a\notin\Delta_{F\lmk \tau_\zeta^{\omega_2}\rmk}$.
Define the $|\Delta^{\caD_{\omega_1}}|\times |\Delta^{\caD_{\omega_2}}|$-matrix
$Y$ by $Y_{b,\xi}:= s_b^\xi$ if $b\in\Delta_{F\lmk \tau_\xi^{\omega_2}\rmk}$,
and $Y_{b,\xi}:= 0$ if $b\notin\Delta_{F\lmk \tau_\xi^{\omega_2}\rmk}$.
Define the $|\Delta^{\caD_{\omega_2}}|\times |\Delta^{\caD_{\omega_1}}|$-matrix
$N$ by $N_{\zeta,a}:= n_{\zeta,a}$ if $a\in\Delta_{F\lmk \tau_\zeta^{\omega_2}\rmk}$,
and $N_{\zeta,a}:= 0$ if $a\notin\Delta_{F\lmk \tau_\zeta^{\omega_2}\rmk}$.
It is immediate to see the properties in Theorem \ref{main}.

\appendix
\section{Mixed state Approximate Haag duality}\label{mah}
By a $2$-dimensional quantum spin system, we mean a $C^*$-algebra constructed as follows.
We denote the algebra of $d\times d$ matrices by $\Mat_{d}$ with $d\ge 2$.
For each $z\in\bbZ^2$,  let $\caA_{\{z\}}$ be an isomorphic copy of $\Mat_{d}$, and for any finite subset $\Lambda\subset\bbZ^2$, we set $\caA_{\Lambda} = \bigotimes_{z\in\Lambda}\caA_{\{z\}}$.
For finite $\Lambda$, the algebra $\caA_{\Lambda} $ can be regarded as the set of all bounded operators acting on
the Hilbert space $\bigotimes_{z\in\Lambda}{\bbC}^{d}$.
We use this identification freely.
If $\Lambda_1\subset\Lambda_2$, the algebra $\caA_{\Lambda_1}$ is naturally embedded in $\caA_{\Lambda_2}$ by tensoring its elements with the identity. 
For an infinite subset $\Gamma\subset \bbZ^{2}$,
$\caA_{\Gamma}$
is given as the inductive limit of the algebras $\caA_{\Lambda}$ with $\Lambda$, finite subsets of $\Gamma$.
We call $\caA_{\Gamma}$ the quantum spin system on $\Gamma$.
The two-dimensional quantum spin system is the algebra $\caA_{\bbZ^2}$. 
For a subset $\Gamma_1$ of $\Gamma\subset\bbZ^{2}$,
the algebra $\caA_{\Gamma_1}$ can be regarded as a subalgebra of $\caA_{\Gamma}$. 
For $\Gamma\subset \bbR^2$, with a bit of abuse of notation, we write $\caA_{\Gamma}$
to denote $\caA_{\Gamma\cap \bbZ^2}$.
Also, $\Gamma^c$ denotes the complement of $\Gamma$ in $\bbR^2$.

In our framework, regions called cones play an important role.
For
each $\bm a\in \bbR^2$, $\theta\in\bbR$ and $\varphi\in (0,\pi)$,
we set 
\begin{align*}
\Lambda_{\bm a, \theta,\varphi}
:
=&\lmk \bm a+\left\{
t\bm e_{\beta}\mid t>0,\quad \beta\in (\theta-\varphi,\theta+\varphi)\right\}\rmk
\end{align*}
where  $\bm e_{\beta}=(\cos\beta,\sin\beta)$.
We call a set of this form a cone.
We set $|\arg\Lambda|=2\varphi$
and $\bm e_\Lambda:=\bm e_{\theta}$ for $\Lambda=\Lambda_{\bm a, \theta,\varphi}$.
For $\varepsilon>0$, $t\in \bbR$ and $\Lambda=\Lambda_{\bm a, \theta,\varphi}$, $\Lambda_\varepsilon$ 
denotes $\Lambda_\varepsilon=\Lambda_{\bm a, \theta,\varphi+\varepsilon}$,
$\Lambda(t):=\ld +t {\bm e}_\ld$.

\begin{defn}\label{AHdef}
Let $\caA$ be a $2$-dimensional quantum spin system.
Let $\omega$ be a state on $\caA$
with a GNS representation $(\caH,\pi)$.
We say $\omega$ satisfies the approximate Haag duality 
if, for any $\zeta\in (0,\pi)$, $0<\varepsilon<\frac14(\pi-\zeta)$,
there exists an $R_{\zeta,\varepsilon}\ge 0$
and a decreasing function $f_{\zeta, \varepsilon}: \bbR_{\ge 0}\to \bbR_{\ge 0}$
with $\lim_{t\to\infty} f_{\zeta, \varepsilon}(t)=0$
satisfying the following : 
for any cone $\Lambda$ with $\lv\arg \Lambda\rv=2\zeta$,
there exists a unitary $U_{\Lambda,\varepsilon}\in \pi\lmk\caA\rmk''$
such that
\begin{description}
\item[(i)]
\[
\pi(\caA_{\Lambda^c})'\cap \pi\lmk \caA_{\bbZ^2}\rmk''
\subset \lmk U_{\Lambda,\varepsilon}\rmk
\pi\lmk\caA_{\Lambda_{\varepsilon}(-R_{\zeta,\varepsilon})}\rmk''\lmk U_{\Lambda,\varepsilon}\rmk^*
,\]
 and
\item[(ii)]
for any $t\ge 0$, there exists a unitary $U_{\Lambda,\varepsilon, t}\in 
\pi\lmk \caA_{\Lambda_{2\varepsilon}(-t)}\rmk''$
such that
\begin{align}
\lV U_{\Lambda,\varepsilon, t}-U_{\Lambda,\varepsilon}\rV
\le f_{\zeta,\varepsilon}(t).
\end{align}
\end{description}
\end{defn}
\bibliographystyle{alpha}
\bibliography{ST}

\end{document}